\documentclass[]{article}
\usepackage{graphicx} 
\usepackage{amsmath,amsfonts,amssymb,amsthm}
\usepackage{bm}
\usepackage{bbm}
\usepackage{url}

\newcommand{\supp}{\operatorname{supp}}

\usepackage{physics}
\usepackage{tikz}
\usetikzlibrary{decorations.pathmorphing}
\usetikzlibrary{decorations.markings,arrows}
\usetikzlibrary{math}
\usetikzlibrary{hobby}
\usepackage[maxcitenames=1,backend=biber,style=chem-angew,sorting=nyt,articletitle=true,doi=false,isbn=false,url=false,eprint=true]{biblatex}
\usepackage[colorinlistoftodos,textsize=tiny]{todonotes}
\usepackage{lineno}
\usepackage[colorlinks=true,linkcolor=red,citecolor=green]{hyperref}

\newtheorem{theorem}{Theorem}
\newtheorem{proposition}{Proposition}

\theoremstyle{definition}
\newtheorem{remark}{Remark}
\theoremstyle{definition}

\usepackage{enumitem}

\title{The Lorentz Gas in a Mean-Field Potential: Weak Coupling and Diffusive Regime}
\author{Dominik Nowak}
\date{\today}

\begin{document}

\maketitle

\begin{abstract}
    We investigate the diffusive scaling of the Lorentz gas in the presence of an external force of mean-field type. In the weak coupling regime and for diffusive time scales, the test particle's law converges to the probability density satisfying the heat equation. The diffusion coefficient of the heat equation is given by the Green-Kubo relation.
\end{abstract}

\section{Introduction}
Let us consider the motion of a test particle of unit mass through a random configuration of spatially fixed obstacles in $d \geq 2$ dimensions. 
For convenience, we restrict ourselves to spherical obstacles with radius one, which are distributed according to a Poisson law with intensity $\mu > 0$. 
This particular toy model goes by the name of \textit{Lorentz gas} \cite{lorentz_motion_1905}.

Depending on the interaction of the test particle with the obstacles, we expect the Lorentz gas to model different physical systems.
In fact, from the literature it is well-known that under suitable rescaling of space $x$, time $t$, interactions $\phi$ and density of obstacles, we can derive different effective equations starting from the particle system. 

An intuition for a certain kinetic behaviour can be deduced from the test particle's mean free path. In the case of the low-density limit, the mean free path has a macroscopic order of magnitude, so that the system models a rarefied gas. Lorentz himself \cite{lorentz_motion_1905} conjectured the linear Boltzmann equation to describe the emerging dynamics. In 1969, Giovanni Gallavotti \cite{gallavotti_divergences_1969,gallavotti_rigorous_nodate} provided the first rigorous proof for Lorentz's claim, deriving a constructive approach to show that the Lorentz gas for hard spheres is described by the linear Boltzmann equation in the low-density limit. Some years later Spohn \cite{spohn_lorentz_1978} improved Gallavotti's result in showing the convergence of path measures of the mechanical system to a random flight process. Further results followed in the paper by Boldrighini et al. \cite{boldrighini_boltzmann_1983}, where the authors showed that the limiting linear Boltzmann equation holds for typical configurations of obstacles. Later, in 1999, Desvillettes and Pulvirenti derived the linear Boltzmann equation for long range forces by modelling the interactions of the obstacles by a truncated power law potential. The low-density limit was taken simultaneously as the truncation was sent to infinity. The technical challenge was in proving the Markovianity of the limiting process, so that the authors had to obtain explicit error bounds \cite{desvillettes_linear_1999}. Further progress on the derivation of the linear Boltzmann equation in the presence of long range interactions from particle systems was made by Ayi in \cite{ayi_newtons_2017}. In this particular paper, Ayi considered the motion of a tagged particle in a gas close to global equilibrium, where the particles interact through a long range potential. The novelty in \cite{ayi_newtons_2017} is that no truncation was imposed, and therefore collisions with the potential's infinite tail had to be dealt with.

The same idea of Gallavotti's proof allowed Desvillettes and Ricci \cite{desvillettes_rigorous_2001} to derive the linear Landau equation from the Lorentz gas in the weak coupling regime. Compared to the low-density limit, the different scaling gives rise to \textit{many} but \textit{soft} collisions with the obstacles and can be thought as a model describing dense plasmas. For \textit{higher} densities, the so-called weak coupling limit is reached. The same physical intuition suggests a diffusion of the velocity process on the kinetic energy sphere, which arises from the mentioned \textit{multiple soft} collisions. Indeed in the weak coupling regime, Kesten and Papanicolau \cite{kesten_limit_1980} showed that the velocity process converges to a Brownian motion on a sphere of constant velocity in three spatial dimensions and the diffusion coefficient depends on the obstacle's interaction potential. In \cite{durr_asymptotic_1987} the authors were able to obtain the result from \cite{kesten_limit_1980} in $\mathbb{R}^2$. Finally, Kirkpatrick derived the linear Landau equation from the mechaical system in the weak coupling limit \cite{kirkpatrick_rigorous_2009}. 

As the linear Boltzmann and Landau equation is only valid for kinetic time scales, we expect other governing equations describing the system’s behaviour in a \textit{suitable hydrodynamic limit}. In principle, a hydrodynamic equation can be obtained from particle systems, either by first performing the kinetic limit and then considering longer time scales, or by directly applying the hydrodynamic scaling. The latter method is usually more challenging. An example for the limit consisting of two steps is the work by Basile et al. \cite{basile_diffusion_2014}, where the authors derived first the linear Landau equation from the Lorentz gas and then obtained the heat equation by studying longer time scales. Similarly, Bodineau et al. proved by the intermediate kinetic formulation that the Brownian motion is a \textit{fast relaxation limit} for the system of hard spheres \cite{bodineau_brownian_2016}. Although previous works as for example \cite{van_beijeren_equilibrium_1980,lebowitz_steady_1982} had already justified the validity of the linear Boltzmann equation, the proof by Bodineau et al. provides a quantitative rate of convergence, which was essential in obtaining the diffusive limit in \cite{bodineau_brownian_2016}. In contrast to these two step limits, a typical work considering directly longer times is by Komorowski and Ryzhik \cite{komorowski_diffusion_2006}, where the authors started from a random Hamiltonian flow and proved the convergence of the associated spatial process to the standard Brownian motion. A similar result was obtained in \cite{erdos_quantum_2008} starting from the quantum Lorentz gas in order to derive the heat equation. These mentioned results were derived in the weak coupling limit, whereas Lutsko and Tóth showed an \textit{invariance principle} for the Lorentz gas in the low-density limit using a probabilistic coupling method \cite{lutsko_invariance_2020}. Of course, also the fully non-linear Boltzmann equation was studied in order to derive hydrodynamic limits. It is hard to mention all the important contributions, so that we refer to the chapter by Golse \cite{golse_chapter_2005} and the summery in the book of Saint-Raymond \cite{saint-raymond_boltzmann_2009} for an overview.

In this sense, the Lorentz gas is particularly interesting from a mathematical perspective, since it allows justifying various effective theories under certain scaling limits. Although the derivation of the Landau equation from particle systems is only understood through a truncated BBGKY hierarchy \cite{bobylev_particle_2013,velazquez_non-markovian_2018} or in the linearised version \cite{le_bihan_long_2024}, a careful analysis of this particular toy model could give further insights in understanding emergent phenomena for the non-linear case as well as its hydrodynamic limit.

\section{Model and Notation} \label{sec:ModelAndNotation}
We introduce a small parameter $\varepsilon >0$ representing the ratio between the micro and macroscopic dynamics. 
For $\alpha \in (0,1/2)$, we consider the weak coupling regime by rescaling space, time and the interaction potential according to 
\begin{equation}\label{eq:ScalingWC}
    x \to \varepsilon x, \qquad t \to \varepsilon t, \qquad \phi \to \varepsilon^\alpha \phi
\end{equation} 
and let the density of obstacles diverge as
\begin{equation}
    \mu_\varepsilon := \varepsilon^{-d+1-2\alpha}.
\end{equation}
In this particular scaling, the obstacles change size and their radius becomes $\varepsilon$. 
Therefore, the probability $ \mathbb{P}^\varepsilon_{N,\Sigma}$ of finding exactly $N$ obstacle centres $\bm{c}_N =(c_1,...,c_N)$ inside a bounded measurable region $\Sigma \subset \mathbb{R}^d$ is given by
\begin{equation} \label{eq:ProbMeasurePois}
    \mathbb{P}^\varepsilon_{N,\Sigma}(\dd \bm{c}_N) = e^{-\mu_\varepsilon \abs{\Sigma}} \frac{\mu_\varepsilon^{ N}}{N !} \dd \bm{c}_N,
\end{equation}
where $\dd \bm{c}_N$ is an abbreviation for $\dd c_1 ... \dd c_N$ and $\abs{\Sigma}$ denotes the Lebesgue measure of $\Sigma$. 

We now introduce the interaction of the test particle with the configuration of obstacles. Let  $U$ be a radially symmetric function satisfying
\begin{enumerate}
    \item $U \in C(\mathbb{R}^d) \cap W^{2,\infty} (\mathbb{R}^d)$,
    \item $U(0)>0$ as well as $r \mapsto U(r)$ is strictly decreasing for $r \in [0,1]$ and
    \item $\supp U \subset [0,1]$.
\end{enumerate}
We model the obstacles in the system with a \textit{soft} interaction potential if we centre $U$ at each $c_i \in \Sigma$ and rescale it according to \eqref{eq:ScalingWC}. This gives 
\begin{equation}
    U^\varepsilon_{c_i} (x) = \varepsilon^\alpha U\left(\frac{\abs{x-c_i}}{\varepsilon} \right).
\end{equation}
We go now further and assume that --- in addition to the scattering potential $U^\varepsilon_{c_i}$ --- each obstacle gives rise to a mean-field potential of the following (rescaled) form 
\begin{equation}
        V^\varepsilon_{c_i} (x) =  \frac{1}{\mu_\varepsilon} \Lambda (x-c_i),
\end{equation}
where $\Lambda$ satisfies
\begin{enumerate}
        \item $\Lambda \in C(\mathbb{R}^d) \cap W^{2,\infty}(\mathbb{R}^d)$ as well as 
        \item $\Lambda = \Lambda(x)$ is radially symmetric and non-increasing in $\abs{x}$.
\end{enumerate}
In fact, $V^\varepsilon_{c_i}$ influences the test particle's motion and contributes to the total force field as
\begin{equation} \label{eq:MeanFieldForce}
    F^\varepsilon_{\Lambda} (x) = \frac{1}{\mu_\varepsilon} \sum_{i=1}^N \nabla \Lambda (x-c_i).
\end{equation}

Given the initial position $x_0$ and velocity $v_0$, we are interested in the Hamiltonian flow $\mathtt{U}_{\bm{c}_N,\varepsilon}^{t}$, which defines the trajectory at time $t$ starting from the pair $(x_0,v_0)$. More precisely, 
\begin{equation}
    (x^\varepsilon(t),v^\varepsilon(t)) := \mathtt{U}_{\bm{c}_N,\varepsilon}^{t}(x_0,v_0),
\end{equation}
where $(x^\varepsilon(t),v^\varepsilon(t))$ satisfies the following equations of motion
\begin{equation}
    \begin{cases}
        \dot{x^\varepsilon}(t) = v^\varepsilon(t)\\
        \dot{v^\varepsilon}(t) = - \sum_{i = 1}^N \left \{ \varepsilon^{\alpha -1} \nabla U\left(\frac{\abs{x-c_i}}{\varepsilon} \right) +\frac{1}{\mu_\varepsilon} \nabla \Lambda (x-c_i)\right \}\\
        (x^\varepsilon(0),v^\varepsilon(0)) = (x_0,v_0).
    \end{cases}
\end{equation}

For a given initial probability density $f_0 = f_0(x,v)$, the law $f_\varepsilon$ of the test particle moving through the random configuration of spherical obstacles is given through
\begin{equation} \label{eq: Expectaion Particle Evolution}
    f_\varepsilon(x,v,t) := \mathbb{E}_\varepsilon [f_0(\mathtt{U}_{\bm{c}_N,\varepsilon}^{-t}(x_0,v_0))],
\end{equation}
where the expectation $\mathbb{E}_\varepsilon [\cdot]$ is taken with respect to the measure in \eqref{eq:ProbMeasurePois}.

The work in preparation \cite{nowak_lorentz_2024} shows that if $\alpha \in (0,(d-1)/8)$, then $f_\varepsilon$ converges to $f=f(x,v,t)$ in $L^1(\mathbb{R}^d \times S^{d-1}_\abs{v})$ for any finite time as $ \varepsilon \to 0$, where $f$ satisfies the linear Landau-Vlasov equation
\begin{equation} \label{eq:Limiting Linear Landau}
	(\partial_{t} + v \cdot \nabla_{x} - \nabla_{x} \Phi	 \cdot \nabla_{v})f = B \Delta_{\abs{v}} f.
\end{equation}
$\Phi$ is the limiting mean-field potential given by the following convolution
\begin{equation} 
    \Phi := \mathbbm{1}_\Sigma * \Lambda
\end{equation}
and $ \Delta_{\abs{v}} $ is the Laplace-Beltrami operator defined on the $d$-dimensional sphere with radius $\abs{v}$.
The diffusion coefficient $B$ depends on the \textit{soft} interaction potential of the obstacles and we refer to \cite{desvillettes_rigorous_2001} for the derivation of an explicit formula of $B$ from $U^\varepsilon_{c_i}$.
From now on, we will denote $B \Delta_\abs{v}$ by $\mathcal{L}$ to enhance the readability. 

In the derivation of \eqref{eq:Limiting Linear Landau} from the Lorentz gas it is of main importance to control the test particle's \textit{memory effects}. 
Although in \cite{nowak_lorentz_2024} the error terms describing the mentioned memory effects vanish as $\varepsilon \to 0$ only for finite times, it is evident from the estimates that this holds true even at the hydrodynamic scale.

Indeed, for long time scales, we expect another effective equation to describe the diffusive behaviour of the system.
Instead of the diffusion in velocity generated by $\Delta_{\abs{v}}$ (see equation \eqref{eq:Limiting Linear Landau}), we expect a randomisation in space. 
The proof is straightforward and we will follow the existing literature by \cite{basile_diffusion_2014} as well as \cite{esposito_chapter_2005}, where the authors were able to derive the heat equation starting from the linear Boltzmann equation. 
The novelty in this short treatise is the presence of the additional mean-field background, which does not appear in the limiting equation. 
Furthermore, the smoothness of the scattering potential allows a strong divergence of the scaling parameter $\eta_\varepsilon$ as $\varepsilon \to 0$, which relates the kinetic to the hydrodynamic description.

This article is organised as follows: In the next section, we motivate the diffusive scaling, which allows studying the hydrodynamic limit and state the main result.
We conclude the paper with section \ref{sec:ProofTHM}, which is devoted to the proof of Theorem \ref{thm:HeatEquation}.
\section{Result}
Let us start from the Lorentz gas introduced in section \ref{sec:ModelAndNotation}. In \cite{nowak_lorentz_2024}, the authors showed that if the intensity of the Poisson distribution equals $\mu_\varepsilon = \varepsilon^{-d+1-2\alpha}$ as well as in the hyperbolic scaling:
\begin{equation}
    x \to \varepsilon x, \qquad t \to \varepsilon t
\end{equation}
$f_\varepsilon$ in \eqref{eq: Expectaion Particle Evolution} and $h_\varepsilon$ satisfying the linear Boltzmann-Vlasov equation
\begin{equation} \label{eq: Intermediate Hyperbolic Boltzmann}
    (\partial_t + v \cdot \nabla_x - \nabla_x \Phi  \cdot \nabla_v) h_\varepsilon = L h_\varepsilon,
\end{equation}
have the same asymptotic limit (in the absence of \textit{pathological} events) as $\varepsilon \to 0$. Here $L$ is the linear collision operator: 
\begin{equation} \label{eq: Definition Collision integral}
    L f := \varepsilon^{-2 \alpha}  \int_{-1}^1 \abs{v}\left \{f(v') - f(v)\right\} \dd \rho
\end{equation}
and $\rho$ denotes the impact parameter as the test particle collides with an obstacle. The post-collision velocity $v'$ is given through the scattering vector $\omega$ associated to the soft interaction potential as follows:
\begin{equation}
    v' := v - 2 (\omega \cdot v) \omega.
\end{equation}
Finally, it is worth mentioning that the prefactor in front of the integral in \eqref{eq: Definition Collision integral} is equal to the inverse of the mean free path.

We are interested in the transport of mass density, so that we take the integral with respect to $v$ of \eqref{eq: Intermediate Hyperbolic Boltzmann}. 
Let us for the moment neglect any normalisation constants in front of the integrals.
First, we realise that 
\begin{equation}
    \int \nabla_x \Phi  \cdot \nabla_v h_\varepsilon \dd v = 0
\end{equation}
by the divergence theorem.
If we further define $\varrho_\varepsilon := \int f_\varepsilon \dd v$ and apply Fubini's theorem, we end up with
\begin{equation} \label{eq: Integrated hyp Boltzmann}
    \partial_t \varrho_\varepsilon + \int v \cdot \nabla_x h_\varepsilon \dd v = 0  = L \int h_\varepsilon \dd v = \int L h_\varepsilon \dd v.
\end{equation}
Indeed, equation \eqref{eq: Integrated hyp Boltzmann} implies that $L h_\varepsilon = 0$ almost everywhere and therefore the limiting function $h$ of \eqref{eq: Intermediate Hyperbolic Boltzmann} must also satisfy $L h = 0$. This means that $h$ has to be proportional to the limiting mass density $\varrho := \lim_{\varepsilon \to 0} \varrho_\varepsilon$ and cannot be a function of the velocity. A simple symmetry argument yields further
\begin{equation}
    \int v \cdot \nabla_x h \dd v = 0,
\end{equation}
since the integrand is odd in $v$.
Finally, we obtain --- as the limiting equation of \eqref{eq: Integrated hyp Boltzmann} --- that
\begin{equation} \label{eq: hyperbolic scaling rho}
    \partial_t \varrho =0
\end{equation}
and therefore we expect no hydrodynamics emerging in this particular scaling. 

In order to experience a hydrodynamic phenomenon, we have to consider the additional diffusive scaling and higher densities:
\begin{equation} \label{eq: Diffusive Scaling Intermediate Boltzmann}
     t \to \eta_\varepsilon^{-\delta}  t, \qquad \mu_\varepsilon = \varepsilon^{-d+1-2 \alpha} \eta_\varepsilon^{\delta}.
\end{equation}
The new time scale and density imply that $\Tilde{h}_\varepsilon(x,v,t) := h_\varepsilon(x,v,t \eta_\varepsilon^{-\delta})$ satisfies
\begin{equation} \label{eq: Diffusive scaling Boltzmann}
    (\partial_t + \eta_\varepsilon^\delta v \cdot \nabla_x - \eta_\varepsilon^\delta \nabla_x \Phi  \cdot \nabla_v) \Tilde{h}_\varepsilon = \eta_\varepsilon^{2\delta} L \Tilde{h}_\varepsilon.
\end{equation}
Considering now \eqref{eq: Diffusive Scaling Intermediate Boltzmann} and the previous discussion, we arrive at the following result:
\begin{theorem}\label{thm:HeatEquation}
    Let $f_0 \in C_0(\mathbb{R}^d \times \mathbb{R}^d)$ be a compactly supported initial probability density. Suppose that $f_0$ has two bounded derivatives with respect to $x$ and $v$, then we prove that
    \begin{equation} \label{eq: Reaching equilibrium Particle Density}
        \lim_{\varepsilon \to 0} f_{\varepsilon}(x,v,t) = \langle f_0 \rangle := K \int_{S^{d-1}_\abs{v}}  f_0(x,v)\dd v
    \end{equation}
    for all $t \in (0,T]$ with $T >0$ in $L^2(\mathbb{R}^d\times S^{d-1}_\abs{v} )$ and $K^{-1} := \frac{2 \pi^{d/2}}{\Gamma(d/2) \abs{v}^{d-1}}$. 
    Furthermore, if we investigate higher densities $\mu_\varepsilon =\varepsilon^{-d+1-2\alpha} \eta_\varepsilon^\delta$ as well as define $F_{\varepsilon}(x,v,t) := f_\varepsilon(x,v,t \eta_\varepsilon^{-\delta})$, with $\delta >0$ we can show, as long as
    \begin{equation}
        \varepsilon^{d-1-8 \alpha} \eta^{4 \delta}_\varepsilon \xrightarrow[]{\varepsilon \to 0} 0,
    \end{equation}
    that for any $t \in [0,T)$
    \begin{equation}
        \norm{F_\varepsilon(x,v,t) - \varrho(x,t)}_{L^2(\mathbb{R}^d \times S^{d-1}_\abs{v})} \xrightarrow[]{\varepsilon \to 0} 0,
    \end{equation}
    where $\varrho$ is the solution of the heat equation
    \begin{equation} \label{eq:HeatEquation}
        \begin{cases}
            \partial_t \varrho - D \Delta \varrho = 0\\
            \varrho(x,0) = \langle f_0 \rangle.
        \end{cases}
    \end{equation}
    The diffusion coefficient $D$ is given by the Green-Kubo relation:
    \begin{equation} \label{eq:DiffusionCoefficient}
        D= -K\int_{S^{d-1}_\abs{v}}  v \cdot \mathcal{L}^{-1} v \dd v =  \int_0^\infty \mathbb{E} [v \cdot V_t(v)] \dd t,
    \end{equation}
    where $V_t(v)$ is the stochastic process associated to the generator $\mathcal{L}$ at time $t$ and starting from $v$. The expectation in \eqref{eq:DiffusionCoefficient} is taken with respect to the measure of the sphere $S^{d-1}_\abs{v}$.
\end{theorem}
\begin{remark}
    Equation \eqref{eq: Reaching equilibrium Particle Density} summarises the result derived in \eqref{eq: hyperbolic scaling rho}. Additionally,
    Theorem \ref{thm:HeatEquation} tells us that $f_\varepsilon$ converges \textit{quickly} to the system's equilibrium $\langle f_0 \rangle$. In order to notice a non-trivial behaviour, we must consider the additional scaling of time $ t \to t \eta_\varepsilon^{-\delta}$ as well as higher densities $\mu_\varepsilon = \varepsilon^{-d+1-2\alpha} \eta_\varepsilon^\delta$, where $\delta > 0$.
\end{remark}
\section{Proof of Theorem \ref{thm:HeatEquation}} \label{sec:ProofTHM}
From now on, we will fix $\delta > 0$ and denote any positive constant independent of $t$ by $C>0$. 
Instead of studying directly \eqref{eq: Diffusive scaling Boltzmann}, it is more convenient to introduce the following Cauchy problem associated to the linear Landau equation:
\begin{equation} \label{eq:RescaledLinearLandauCauchy}
    \begin{cases}
        (\partial_t + \eta^\delta v \cdot \nabla_x - \eta^\delta \nabla_x \Phi \cdot \nabla_v)g_\eta(x,v,t) = \eta^{2 \delta}\mathcal{L} g_\eta(x,v,t)\\
        g_\eta(t=0) = f_0(x,v),
    \end{cases}
\end{equation}
with $\eta \equiv \eta_\varepsilon$ in order to enhance the readability. Furthermore, we remind ourselves that $\mathcal{L} := B \Delta_{\abs{v}}$. We note immediately that the solution $g_\eta$ of \eqref{eq:RescaledLinearLandauCauchy} inherits the regularity of $f_0$ and gains regularity with respect to the transverse component of $v$ as a consequence of the Laplace-Beltrami operator.

For the particular Cauchy problem in \eqref{eq:RescaledLinearLandauCauchy} we show
\begin{proposition}\label{prop:ConvergenceToAverage}
    Let us define $\langle g_\eta \rangle := K \int_{S^{d-1}_\abs{v}}  g_\eta(x,v,t)\dd v$ with $g_\eta$ satisfying \eqref{eq:RescaledLinearLandauCauchy}. Under the assumptions of Theorem \ref{thm:HeatEquation}, we prove for all $t \in (0,T]$ that
    \begin{equation}
        g_\eta - \langle g_\eta \rangle \xrightarrow[]{\eta \to \infty} 0 \quad \text{in} \quad L^2(\mathbb{R}^d \times S^{d-1}_\abs{v}).
    \end{equation}
    Furthermore, if we define $t_\eta := \frac{1}{\eta^\omega}$ with $\omega >2\delta$, we deduce that
    \begin{equation}
        g_\eta(t_\eta) - \langle f_0\rangle \xrightarrow[]{\eta \to \infty} 0 \quad \text{in} \quad L^2(\mathbb{R}^d \times S^{d-1}_\abs{v}),
    \end{equation}
    where $\langle f_0\rangle := K\int_{S^{d-1}_\abs{v}}  f_0(x,v)\dd v$.
\end{proposition}
\begin{proof}
    Let us introduce $R_\eta := g_\eta - \langle g_\eta \rangle$. 
    By applying \eqref{eq:RescaledLinearLandauCauchy} to $R_\eta$, we obtain
    \begin{equation}
        \begin{aligned}
            (\partial_t + \eta^\delta v \cdot \nabla_x - \eta^\delta \nabla_x \Phi \cdot \nabla_v) R_\eta &= (\partial_t + \eta^\delta v \cdot \nabla_x - \eta^\delta \nabla_x \Phi \cdot \nabla_v)(g_\eta - \langle g_\eta \rangle)\\
            &= \eta^{2\delta} \mathcal{L} g_\eta - (\partial_t + \eta^\delta v \cdot \nabla_x - \eta^\delta \nabla_x \Phi \cdot \nabla_v) \langle g_\eta \rangle\\
            &= \eta^{2\delta} \mathcal{L}( g_\eta - \langle g_\eta \rangle) - (\partial_t + \eta^\delta v \cdot \nabla_x) \langle g_\eta \rangle,
        \end{aligned}
    \end{equation}
    where the last equality arises by noting that $\langle g_\eta \rangle$ is not a function of $v$. The expression from above reduces to
    \begin{equation}
        (\partial_t + \eta^\delta v \cdot \nabla_x - \eta^\delta \nabla_x \Phi \cdot \nabla_v) R_\eta = \eta^{2\delta}\mathcal{L} R_\eta + \varphi
    \end{equation}
    with
    \begin{equation}
        \begin{aligned}
            \varphi&=-(\eta^\delta v \cdot \nabla_x \langle g_\eta \rangle + \partial_t \langle g_\eta \rangle) = -\left(\eta^\delta v \cdot \nabla_x \langle g_\eta \rangle + \int_{S^{d-1}_\abs{v}}\partial_t g_\eta \dd v\right)\\
            &= -\eta^\delta v \cdot \nabla_x \langle g_\eta \rangle +K\int_{S^{d-1}_\abs{v}}\eta^\delta (v \cdot \nabla_x - \nabla_x \Phi \cdot \nabla_v)g_\eta\dd v - K \eta^{2\delta}\int_{S^{d-1}_\abs{v}} \mathcal{L} g_\eta \dd v \\
            &=\eta^\delta \left(K\int_{S^{d-1}_\abs{v}}v \cdot \nabla_x g_\eta\dd v - v \cdot \nabla_x \langle g_\eta \rangle \right).
        \end{aligned}
    \end{equation}
    In order to obtain the last line from above, we applied the divergence theorem to $\langle \mathcal{L} g_\eta \rangle$. More precisely, let $\hat{n}$ be the the outward pointing unit normal and $\dd \sigma$ the surface measure on $S^{d-1}_\abs{v}$, then
    \begin{equation} \label{eq:DivTheoremLandauOperator}
        \int_{S^{d-1}_\abs{v}} \mathcal{L}g_\eta \dd v = \int_{\partial S^{d-1}_\abs{v}} \nabla_\abs{v} g_\eta \cdot \hat{n} \dd \sigma = 0,
    \end{equation}
    since $\partial S^{d-1}_\abs{v} = \varnothing$. 
    We find in a similar calculation that the external force $-\eta \nabla_x \Phi \cdot \nabla_v g_\eta$ vanishes too. 
    This leads us to an estimate on $\varphi$:
    \begin{equation}
        \sup_{t \leq T} \norm{\varphi}_{L^2}  \leq \sup_{t \leq T}  C \eta^\delta \norm{\nabla_x g_\eta }_{L^2}  \leq \eta^\delta C T.
    \end{equation}
    From now on, let us refer to the standard inner product in $L^2$ as $(\cdot,\cdot)$. Then, $R_\eta$ satisfies the following ODE
    \begin{equation}\label{eq:ODE Reta}
        \begin{aligned}
            \frac{1}{2} \frac{\dd }{\dd t} \norm{R_\eta(t)}^2_{L^2} &= \eta^{2\delta} (R_\eta,\mathcal{L}R_\eta) + (R_\eta,\varphi) \\
            &\leq -\eta^{2\delta}(R_\eta,-\mathcal{L}R_\eta) + \norm{R_\eta}_{L^2} \norm{\varphi}_{L^2}\\
            &\leq - \eta^{2\delta} \lambda \norm{R_\eta}^2_{L^2} + \norm{R_\eta}_{L^2} \norm{\varphi}_{L^2}.
        \end{aligned}
    \end{equation}
    The last inequality is justified by introducing the smallest eigenvalue $\lambda>0$ corresponding to the positive operator $-\mathcal{L}$.
    We solve \eqref{eq:ODE Reta} and find
    \begin{equation}
        \begin{aligned}
            \norm{R_\eta(t)}_{L^2} &\leq e^{-\lambda\eta^{2\delta} t} \norm{R_\eta(0)}_{L^2} + \int_0^t  e^{-\lambda \eta^{2\delta} (t-s)} \norm{\varphi (s)}_{L^2} \dd s\\
            &\leq e^{-\lambda\eta^{2\delta} t} \norm{R_\eta(0)}_{L^2} + \frac{C}{\lambda \eta^\delta} \left(1-e^{-\lambda \eta^{2\delta} t}\right).
        \end{aligned}
    \end{equation}
    Therefore, $\norm{R_\eta(t)}_{L^2} \xrightarrow[]{\eta \to \infty}0$ for all $t \in (0,T]$ implying  
    \begin{equation} \label{eq:Item1 PropConvergenceToAverage}
        g_\eta - \langle g_\eta \rangle \xrightarrow[]{\eta \to \infty} 0 \quad \text{in} \quad L^\infty ((0,T];L^2(\mathbb{R}^d \times S^{d-1}_{\abs{v}})).
    \end{equation}
    In order to proceed, we consider
    \begin{equation}
        \begin{aligned}
            \frac{1}{2} \frac{\dd}{\dd t} \norm{g_\eta (t) - f_0}^2_{L^2} &= (g_\eta - f_0, (\partial_t + \eta^\delta v \cdot \nabla_x - \eta^\delta \nabla_x \Phi \cdot \nabla_v)(g_\eta-f_0))\\
            &= (g_\eta - f_0,(\partial_t + \eta^\delta v \cdot \nabla_x - \eta^\delta \nabla_x \Phi \cdot \nabla_v)g_\eta)\\
            &- (g_\eta - f_0,( \eta^\delta v \cdot \nabla_x - \eta^\delta \nabla_x \Phi \cdot \nabla_v)f_0),
        \end{aligned}
    \end{equation}
    where we noted that $f_0$ is not a function of time. This leads us to
    \begin{equation} \label{eq:ClosenessTof0 ODE}
        \begin{aligned}
            &\frac{1}{2} \frac{\dd}{\dd t} \norm{g_\eta (t) - f_0}^2_{L^2} \\
            &= (g_\eta - f_0 , \eta^{2\delta} \mathcal{L} g_\eta) - \eta^\delta (g_\eta - f_0 ,  v \cdot \nabla_x f_0) + \eta^\delta(g_\eta - f_0 , \nabla_x \Phi \cdot \nabla_v f_0)\\
            &\leq \eta^{2\delta} (g_\eta - f_0 , \mathcal{L}f_0)- \eta^\delta (g_\eta - f_0 ,  v \cdot \nabla_x f_0) + \eta^\delta (g_\eta - f_0 , \nabla_x \Phi \cdot \nabla_v f_0)\\
            &\leq \norm{g_\eta - f_0}_{L^2} \left(\eta^\delta \abs{v} \norm{\nabla_x f_0}_{L^2}+ \eta^\delta \norm{\nabla_x \Phi}_{L^2}\norm{\nabla_v f_0}_{L^2} + \eta^{2\delta} \norm{\mathcal{L}f_0}_{L^2}\right).
        \end{aligned}
    \end{equation}
    We solve the ODE in \eqref{eq:ClosenessTof0 ODE} and deduce that
    \begin{equation}
        \norm{g_\eta(t) - f_0}_{L^2} \leq t \left(\eta^\delta \abs{v} \norm{\nabla_x f_0}_{L^2}+ \eta^\delta \norm{\nabla_x \Phi}_{L^2}\norm{\nabla_v f_0}_{L^2} + \eta^{2\delta} \norm{\mathcal{L}f_0}_{L^2}\right),
    \end{equation}
    where we have used $g_\eta(0) = f_0$ to determine the integration constant. 
    As an intermediate result, we obtain
    \begin{equation}
        g_\eta(t_\eta) - f_0 \xrightarrow[]{\eta \to \infty} 0 \quad \text{in} \quad L^2(\mathbb{R}^d \times S^{d-1}_{\abs{v}}).
    \end{equation}
    Finally, this yields
    \begin{equation}
        \begin{aligned}
            \norm{g_\eta(t_\eta) - \langle f_0 \rangle}_{L^2} &\leq \norm{g_\eta(t_\eta) - \langle g_\eta (t_\eta)\rangle}_{L^2} + \norm{ \langle g_\eta (t_\eta)\rangle - \langle f_0 \rangle}_{L^2}\\
            &\leq \sup_{0<t \leq T} \norm{g_\eta(t) - \langle g_\eta (t)\rangle}_{L^2}+ \norm{ \langle g_\eta (t_\eta) - f_0 \rangle}_{L^2}\\
            &\leq \sup_{0<t \leq T} \norm{g_\eta(t) - \langle g_\eta (t)\rangle}_{L^2}+ \langle \norm{  g_\eta (t_\eta) - f_0 }_{L^2} \rangle
        \end{aligned}
    \end{equation}
    by the triangle inequality, the linearity of $\langle \cdot \rangle$ as well as \eqref{eq:Item1 PropConvergenceToAverage}. The calculation above proves the second claim of Proposition \ref{prop:ConvergenceToAverage}:
    \begin{equation}
        \norm{g_\eta (t_\eta) - \langle f_0 \rangle}_{L^2(\mathbb{R}^d \times S^{d-1}_\abs{v})} \xrightarrow{\eta \to \infty} 0 .
    \end{equation}
\end{proof}
We continue proving Theorem \ref{thm:HeatEquation} by using the result below. The next proposition tells us that the solution to \eqref{eq:RescaledLinearLandauCauchy} is asymptotically close to $\varrho$ satisfying the heat equation in \eqref{eq:HeatEquation}.
\begin{proposition} \label{prop: Hilbert Expansion}
    Let $g_\eta$ be the solution to \eqref{eq:RescaledLinearLandauCauchy}. Then,  
    \begin{equation}
        g_\eta - \varrho \xrightarrow[]{\eta \to \infty} 0 \quad \text{in} \quad L^\infty ([0,T];L^2(\mathbb{R}^d \times S^{d-1}_{\abs{v}})),
    \end{equation}
    where $\varrho$ satisfies \eqref{eq:HeatEquation}. The diffusion coefficient $D$ is given by the Green-Kubo relation in \eqref{eq:DiffusionCoefficient}.
\end{proposition}
\begin{proof}
    We start by considering the truncated Hilbert expansion of $g_\eta$, which is given through
    \begin{equation} \label{eq: Truncated Hilbert}
        g_\eta(x,v,t) = g^{(0)} (x,v,t) + \frac{1}{\eta^\delta} g^{(1)}(x,v,t) + \frac{1}{\eta^{2\delta}} g^{(2)}(x,v,t) + \frac{1}{\eta^\delta}R_\eta.
    \end{equation}
    $R_\eta$ denotes the reminder of the series.
    We insert the expansion \eqref{eq: Truncated Hilbert} into \eqref{eq:RescaledLinearLandauCauchy} and obtain
    \begin{equation} \label{eq:HilbertInPDE}
        \begin{aligned}
            &(\partial_t + \eta^\delta v \cdot \nabla_x - \eta^\delta \nabla_x \Phi \cdot \nabla_v)g_\eta\\
            &= (\partial_t + \eta^\delta v \cdot \nabla_x - \eta^\delta \nabla_x \Phi \cdot \nabla_v)\left(g^{(0)} + \frac{1}{\eta^\delta} g^{(1)}+ \frac{1}{\eta^{2\delta}} g^{(2)} + \frac{1}{\eta^\delta}R_\eta\right)\\
            &= \eta^{2\delta} \mathcal{L} g_\eta = \eta^{2\delta} \mathcal{L} \left(g^{(0)} + \frac{1}{\eta^\delta} g^{(1)}+ \frac{1}{\eta^{2\delta}} g^{(2)} + \frac{1}{\eta^\delta}R_\eta\right).
        \end{aligned}
    \end{equation}
    If we organise \eqref{eq:HilbertInPDE} with respect to the powers of $\eta$ and define $A_\eta(t):= \partial_t g^{(1)} + \frac{1}{\eta^\delta}\partial_t g^{(2)}+v\cdot\nabla_x g^{(2)} - \nabla_x \Phi \cdot \nabla_v g^{(2)}$, we arrive at the following set of equations
    \begin{enumerate}[label={(\roman*)}]
        \item $\mathcal{L} g^{(0)} = 0$,
        \item \label{item:two} $v \cdot \nabla_x g^{(0)} -\nabla_x \Phi \cdot \nabla_v g^{(0)} = \mathcal{L} g^{(1)}$,
        \item $\partial_t g^{(0)}+ v \cdot \nabla_x g^{(1)}-\nabla_x \Phi \cdot \nabla_v g^{(1)} = \mathcal{L}g^{(2)}$ as well as
        \item $(\partial_t + \eta^\delta v \cdot \nabla_x - \eta^\delta \nabla_x \Phi \cdot \nabla_v) R_\eta = \eta^{2\delta} \mathcal{L}R_\eta - A_\eta(t)$.
    \end{enumerate}
    Clearly, the functions satisfying $\mathcal{L}g^{(0)} = 0$ must belong to the null space of the Laplace-Beltrami operator. As a consequence of the periodicity of the sphere, we conclude that the only functions, which are \textit{harmonic} with respect to $\mathcal{L}$, are constants in velocity space.
    
    Let us now consider the second equation from above. Item \ref{item:two} has only a solution if the left hand side belongs to 
    \begin{equation}
    	(\mathrm{ker} (\mathcal{L}))^{\perp} := \left\{h \in L^{2}(S_{\abs{v}}^{d-1} ): \int_{S^{d-1}_{\abs{v}}} h(v) \dd v = 0 \right\}.
    \end{equation}
        Indeed, the second equation above yields
    \begin{equation}
        \int_{S^{d-1}_\abs{v}}  \left( v \cdot \nabla_x -\nabla_x \Phi \cdot \nabla_v \right) g^{(0)}  \dd v =  \int_{S^{d-1}_\abs{v}} v \cdot \nabla_x  g^{(0)}  \dd v =0,
    \end{equation}
    since $g^{(0)} $ does not depend on the velocity and therefore $ v \cdot \nabla_x  g^{(0)} $ is an odd function in $v$. At this point we introduce the pseudo inverse $\mathcal{L}^{-1}$ and deduce
    \begin{equation} \label{eq:Identity g1}
        g^{(1)} = \mathcal{L}^{-1} v \cdot \nabla_x g^{(0)} . 
    \end{equation}
    For the third equation of the list, we find that the right hand side gives
    \begin{equation} \label{eq: Integral v Lg2}
        \int_{S^{d-1}_\abs{v}}\mathcal{L}g^{(2)} \dd v = \int_{\partial S^{d-1}_\abs{v}} \nabla_\abs{v} g^{(2)} \cdot \hat{n} \dd \sigma  = 0
    \end{equation}
    by the same argument as in \eqref{eq:DivTheoremLandauOperator}. Integrating the left hand side with respect to $v$ and using the result from above yields
    \begin{equation} \label{eq:Heat g0}
        \partial_t g^{(0)} + K\int_{S^{d-1}_\abs{v}}v \cdot \nabla_x g^{(1)} \dd v = 0
    \end{equation}
    after applying the divergence theorem to $\nabla_x \Phi \cdot \nabla_v g^{(1)}$.
    We continue by inserting the identity \eqref{eq:Identity g1} into \eqref{eq:Heat g0}:
    \begin{equation}
        \begin{aligned}
            &\partial_t g^{(0)} + K\int_{S^{d-1}_\abs{v}} v \cdot \nabla_x \left \{ \mathcal{L}^{-1}  v \cdot \nabla_x g^{(0)} \right \} \dd v\\
            &= \partial_t g^{(0)} + K\int_{S^{d-1}_\abs{v}} v \cdot  \mathcal{L}^{-1} v \Delta_x g^{(0)}  \dd v =0.
        \end{aligned}
    \end{equation}
    Let us introduce the diffusion coefficient
    \begin{equation} \label{eq:DiffusionIndex-Form}
        D_{ij}:= - K\int_{S^{d-1}_\abs{v}}  v_i \mathcal{L}^{-1} v_j \dd v.
    \end{equation}
    As a direct consequence of symmetry, we note $D_{ij} = D \delta_{ij}$. 
    Furthermore, we see that $D>0$ by the negativity of $\mathcal{L}^{-1}$ and we can obtain the desired Green-Kubo relation from \eqref{eq:DiffusionIndex-Form}:
    \begin{equation}
        \begin{aligned}
            D= -K \int_{S^{d-1}_\abs{v}}  v \cdot \mathcal{L}^{-1} v \dd v &= K\int_{S^{d-1}_\abs{v}} \int_0^\infty v \cdot e^{\mathcal{L} t} v \dd t \dd v\\
            & =K\int_{S^{d-1}_\abs{v}} \int_0^\infty \mathbb{E} [v \cdot V_t(v)] \dd t \dd v\\
            & =  \int_0^\infty \mathbb{E} [v \cdot V_t(v)] \dd t.
        \end{aligned}
    \end{equation}
    Going back to equation \eqref{eq:Heat g0}, we use the definition of $D$ to find
    \begin{equation}
        \partial_t g^{(0)} - D \Delta_x  g^{(0)}  = 0.
    \end{equation}
    From now on we will assume that $g^{(0)}$ satisfies the following initial condition $g_\eta(x,v,t=0) = g^{(0)}(x,v,0)$. It is obvious that $g^{(0)}$ stays in $L^2$, since it satisfies the heat equation. Furthermore, by equation \eqref{eq:Identity g1}, the regularity of $g^{(0)}$ also implies $g^{(1)} \in L^2$.

    In \eqref{eq: Integral v Lg2} we showed that $\int_{S^{d-1}_\abs{v}} \mathcal{L} g^{(2)} \dd v =0$, which allows us to invert $\mathcal{L}$ and find an expression for $g^{(2)}$:
    \begin{equation}
        \begin{aligned}
            g^{(2)}& = \mathcal{L}^{-1} \left(\partial_t g^{(0)}+ v \cdot \nabla_x g^{(1)}-\nabla_x \Phi \cdot \nabla_v g^{(1)} \right)\\
            &= \mathcal{L}^{-1} \left(D \Delta_x  g^{(0)} + v \cdot \mathcal{L}^{-1} v \Delta_x g^{(0)}-\nabla_x \Phi \cdot \nabla_v \left \{\mathcal{L}^{-1} v \cdot \nabla_x g^{(0)} \right \}\right).\\
        \end{aligned}
    \end{equation}
    This equality shows that the $L^2$-norm of $g^{(2)}$ is bounded as well, since we can express each term in the equation above by (suitable derivatives of) $g^{(0)}$.

    The last step consists of controlling the $L^2$-norm of $R_\eta$. We therefore study
    \begin{equation}
        \begin{aligned}
            \frac{1}{2} \frac{\dd}{\dd t} \norm{R_\eta(t)}^2_{L^2} &= - \eta^{2\delta} (R_\eta, - \mathcal{L}R_\eta) - (R_\eta, A_\eta(t))\\
            &\leq - \lambda \eta^{2\delta} \norm{R_\eta}^2_{L^2} + \norm{R_\eta}_{L^2} \norm{A_\eta (t)}_{L^2}
        \end{aligned}
    \end{equation}
    by the same argument as in \eqref{eq:ODE Reta}, which implies
    \begin{equation}
        \frac{\dd}{\dd t} \norm{R_\eta(t)}_{L^2} \leq  \norm{A_\eta (t)}_{L^2}.
    \end{equation}
    We need to find a suitable estimate for $A_\eta$ in $L^2$. As an intermediate result, let us consider
    \begin{equation}
        \begin{aligned}
            \partial_t g^{(1)} &= \partial_t \left(\mathcal{L}^{-1}  v \cdot \nabla_x g^{(0)}\right)\\
            &= \mathcal{L}^{-1} \left( \Dot{v} \cdot \nabla_x g^{(0)} + v \cdot \nabla_x \partial_t g^{(0)}\right)\\
            &=\mathcal{L}^{-1} \left(\Dot{v} \cdot \nabla_x g^{(0)} + v \cdot \nabla_x \left \{D \Delta_x g^{(0)} \right \}\right).
        \end{aligned}
    \end{equation}
    The last line tells us that $\partial_t g^{(1)}$ is indeed bounded in $L^2$. Similarly, we can express $\partial_t g^{(2)}$, $v \cdot \nabla_x g^{(2)}$ as well as $-\nabla_x \Phi \cdot \nabla_v g^{(2)}$ in terms of derivatives with respect to $x$ and $v$ of $g^{(0)}$. Hence, $\norm{A_\eta}_{L^2}$ is uniformly bounded for all $t \in [0,T]$:
    \begin{equation}
        \norm{R_\eta(t)}_{L^2} \leq \int_0^T  \norm{A_\eta(s)}_{L^2} \dd s \leq C T.
    \end{equation}
    This estimate of $\norm{R_\eta(t)}_{L^2}$ shows that the Hilbert expansion in \eqref{eq: Truncated Hilbert} converges to $g^{(0)} \equiv \varrho$ as $\eta \to \infty$ and we conclude the proof.
\end{proof}
\begin{remark}
    It is indeed enough to consider the truncated Hilbert expansion in \eqref{eq: Truncated Hilbert}, since the general ansatz 
    \begin{equation}
        g_\eta = \sum_{k=0}^\infty \eta^{-k \delta} g^{(k)}
    \end{equation}
    would yield the following recursion formula
    \begin{equation}
        \partial_t g^{(n)} + v \cdot \nabla_x g^{(n+1)} - \nabla_x \Phi \cdot \nabla_v g^{(n+1)} = \mathcal{L} g^{(n+2)}, \qquad n \geq 1
    \end{equation}
    and the procedure in the proof could be repeated for any $n$.
\end{remark}
We finalise the claim of Theorem \ref{thm:HeatEquation} if we prove
\begin{proposition} \label{prop: Comparision Operators}
    Let us consider $\Tilde{h}_\varepsilon$ given by \eqref{eq: Diffusive scaling Boltzmann} with initial condition $f_0$. Then, for all $t \in [0,T]$
    \begin{equation}
        \Tilde{h}_\varepsilon - \varrho \xrightarrow[]{\varepsilon \to 0} 0 \quad \text{in} \quad L^2(\mathbb{R}^d \times S^{d-1}_\abs{v}).
    \end{equation}
\end{proposition}
\begin{proof}
    Let $g_{\eta_\varepsilon}$ satisfy \eqref{eq:RescaledLinearLandauCauchy}, whereas $\eta \equiv \eta_\varepsilon$. Then, if we study
    \begin{equation} \label{eq: Difference Equations Boltzmann Landau}
        (\partial_t + \eta_\varepsilon^\delta v \cdot \nabla_x - \eta_\varepsilon^\delta \nabla_x \Phi \cdot \nabla_v)(\Tilde{h}_\varepsilon-g_{\eta_\varepsilon}) = \eta_\varepsilon^{2\delta} ({L} \Tilde{h}_\varepsilon - \mathcal{L}g_{\eta_\varepsilon}),
    \end{equation}
    we obtain
    \begin{equation}
        \begin{aligned}
            \frac{1}{2} \frac{\dd}{\dd t} \norm{\Tilde{h}_\varepsilon-g_{\eta_\varepsilon}}^2_{L^2} &= \eta_\varepsilon^{2\delta} (\Tilde{h}_\varepsilon-g_{\eta_\varepsilon},{L} \Tilde{h}_\varepsilon - \mathcal{L}g_{\eta_\varepsilon})\\
            &= -\eta_\varepsilon^{2\delta} (\Tilde{h}_\varepsilon-g_{\eta_\varepsilon},-{L} [\Tilde{h}_\varepsilon-g_{\eta_\varepsilon}]) + \eta_\varepsilon^{2\delta} (\Tilde{h}_\varepsilon-g_{\eta_\varepsilon}, [{L} - \mathcal{L}]g_{\eta_\varepsilon})\\
            &\leq -\eta_\varepsilon^{2\delta} \lambda_L \norm{\Tilde{h}_\varepsilon-g_{\eta_\varepsilon}}^2_{L^2}+\eta_\varepsilon^{2\delta} \norm{\Tilde{h}_\varepsilon-g_{\eta_\varepsilon}}_{L^2} \norm{[{L} - \mathcal{L}]g_{\eta_\varepsilon}}_{L^2}\\
            &\leq \eta_\varepsilon^{2\delta} \norm{\Tilde{h}_\varepsilon-g_{\eta_\varepsilon}}_{L^2} \norm{[{L} - \mathcal{L}]g_{\eta_\varepsilon}}_{L^2},
        \end{aligned}
    \end{equation}
    where we used the positivity of $-{L}$ with eigenvalue $\lambda_L >0$. The last inequality arises by noting that the first term in the third line is always smaller than zero and will lead to a decay of $\norm{\Tilde{h}_\varepsilon-g_{\eta_\varepsilon}}_{L^2}$. We obtain the final ODE
    \begin{equation}
        \frac{\dd}{\dd t} \norm{\Tilde{h}_\varepsilon-g_{\eta_\varepsilon}}_{L^2} \leq \eta_\varepsilon^{2\delta}  \norm{[{L} - \mathcal{L}]g_{\eta_\varepsilon}}_{L^2}.
    \end{equation}
    Let us recall at this point the definition of the linear collision operator:
    \begin{equation} \label{eq:RescaledBoltzmannOperator}
        {L} g_{\eta_\varepsilon} =  \abs{v} \varepsilon^{-2\alpha} \int_{-1}^1  \left \{ g_{\eta_\varepsilon}(x,v',t) - g_{\eta_\varepsilon}(x,v,t) \right\}\dd \rho.
    \end{equation}
    In the limit of grazing collisions, we can expand the integrand in \eqref{eq:RescaledBoltzmannOperator} for small changes in velocity. More specifically, this expansion results in
    \begin{equation} \label{eq:ExpansionIntegran}
        \begin{aligned}
            g_{\eta_\varepsilon}(v') - g_{\eta_\varepsilon}(v) = &(v'- v) \cdot \nabla_{\abs{v}} g_{\eta_\varepsilon}(v)\\
            &+ \frac{1}{2} (v'-v) \otimes (v'-v) : \nabla_\abs{v} \nabla_\abs{v}g_{\eta_\varepsilon}(v) + \Tilde{R}_{\eta_\varepsilon},
        \end{aligned}
    \end{equation}
    where $\Tilde{R}_{\eta_\varepsilon} = O(\abs{v-v'}^3)$. Inserting \eqref{eq:ExpansionIntegran} into \eqref{eq:RescaledBoltzmannOperator} gives
    \begin{equation}
        \begin{aligned}
            {L} g_{\eta_\varepsilon} =&  \abs{v}  {\varepsilon^{-2\alpha}}\int_{-1}^1   \Big\{ (v-v') \cdot \nabla_\abs{v} g_{\eta_\varepsilon}(v) \\
            &+ \frac{1}{2} (v'-v) \otimes (v'-v) : \nabla_\abs{v} \nabla_\abs{v}g_{\eta_\varepsilon}(v) + \Tilde{R}_{\eta_\varepsilon} \Big\}\dd \rho\\
            &= \abs{v}  {\varepsilon^{-2\alpha}}\int_{-1}^1   \left\{ \frac{1}{2} \Delta_\abs{v} g_{\eta_\varepsilon} \abs{v-v'}^2 + R_{\eta_\varepsilon} \right\}\dd \rho,
        \end{aligned}
    \end{equation}
    where we noted that the odd terms in $(v-v')$ vanish due to symmetry and introduced $R_{\eta_\varepsilon} = O(\abs{v-v'}^4 )$. 
    
    Finally, by using the identity $\abs{v-v'}^2 = 4 \abs{v}^2 \sin^2{\left(\frac{\theta(\rho)}{2} \right)}$, we obtain
    \begin{equation}
        \begin{aligned}
            \abs{v} {\varepsilon^{-2\alpha}}\int_{-1}^1  \abs{v'-v}^2 \dd \rho  &=  \abs{v} {\varepsilon^{-2\alpha}} \int_{-1}^1  4 \abs{v}^2 \sin^2{\left(\frac{\theta(\rho)}{2} \right)} \dd \rho.
        \end{aligned}
    \end{equation}
    If we apply basic scattering theory to the function modelling the interaction potential, we find that $\theta(\rho) \leq C \varepsilon^\alpha$ (see also \cite{nowak_lorentz_2024}). 
    Then, the expression from above coincides by the leading order with $B$ in equation \eqref{eq:Limiting Linear Landau}. 
    We can match $\mathcal{L}$ with ${L}$ and get
    \begin{equation}
        \norm{[{L} - \mathcal{L}]g_{\eta_\varepsilon}}_{L^2} \leq C \varepsilon^{2 \alpha}  \norm{\Delta_\abs{v}^2 g_{\eta_\varepsilon}} \leq C \varepsilon^{2 \alpha} ,
    \end{equation}
    which implies the following uniform bound in time
    \begin{equation} \label{eq:MatchBoltzmannLandau}
        \sup_{t \in [0,T]}\norm{\Tilde{h}_\varepsilon-g_{\eta_\varepsilon}}_{L^2} \leq C \eta_\varepsilon^{2\delta} \varepsilon^{2 \alpha}  T.
    \end{equation}
    Equation \eqref{eq:MatchBoltzmannLandau} tells us that for all $t \in [0,T]$
    \begin{equation}
        \Tilde{h}_\varepsilon-g_{\eta_\varepsilon} \xrightarrow[]{\varepsilon \to 0} 0 \quad \text{in} \quad L^2(\mathbb{R}^d \times S^{d-1}_{\abs{v}}).
    \end{equation}
\end{proof}
\begin{remark}
    In the proof of Proposition \ref{prop: Comparision Operators}, it is evident that we need higher densities. If the density remained the same, the right hand side of \eqref{eq: Difference Equations Boltzmann Landau} would be 
    \begin{equation}
        \eta_\varepsilon^{2 \delta}\left(\frac{L}{\eta_\varepsilon^\delta}\Tilde{h}_\varepsilon - \mathcal{L}g_{\eta_\varepsilon} \right)
    \end{equation}
    and therefore could not coincide with $\mathcal{L}$ in \eqref{eq:RescaledLinearLandauCauchy}.
\end{remark}
We conclude Theorem \ref{thm:HeatEquation} by commenting on the limiting behaviour of the functions introduced above. More precisely, if $\mu_\varepsilon= \varepsilon^{-d+1-2\alpha}\eta_\varepsilon^\delta$, we want to argue that $f_\varepsilon (\eta_\varepsilon^{-\delta} t)$ converges to $\Tilde{h}_\varepsilon(t)$ in $L^2(\mathbb{R}^d \times S^{d-1}_\abs{v})$ for all $t \in [0,T]$. Although the authors in \cite{nowak_lorentz_2024} proved only
\begin{equation} \label{eq:ConvergenceResultChiaraDomi}
    {f}_\varepsilon (t) - {h}_\varepsilon(t) \xrightarrow[]{\varepsilon \to 0} 0 \quad \text{in} \quad L^1(\mathbb{R}^d \times S^{d-1}_\abs{v})
\end{equation}
for any $t \in [0,T]$, the result in \eqref{eq:ConvergenceResultChiaraDomi} still holds for time scales $t \to \eta_\varepsilon^{-\delta} t$ and $\mu_\varepsilon= \varepsilon^{-d+1-2\alpha}\eta_\varepsilon^\delta$ as long as the \textit{critical} error 
\begin{equation}
    \varepsilon^{d-1-8 \alpha} \eta^{4 \delta}_\varepsilon \xrightarrow[]{\varepsilon \to 0} 0
\end{equation}
for a suitable function $\eta_\varepsilon$.
Furthermore, since $f_0$ has compact support under the assumption of Theorem \ref{thm:HeatEquation}, we conclude that ${h}_\varepsilon(\eta_\varepsilon^{-\delta} t)$ converges to $\Tilde{h}_\varepsilon(t)$ for every $t \in [0,T]$ in $L^1(\mathbb{R}^d \times S^{d-1}_\abs{v})$.
This in turn guarantees the asymptotic closeness of ${f}_\varepsilon (\eta_\varepsilon^{-\delta} t)$ and $\Tilde{h}_\varepsilon (t)$.
Again --- by hypothesis --- the functions ${f}_\varepsilon (\eta_\varepsilon^{-\delta} t)$ as well as $\Tilde{h}_\varepsilon (t)$ are uniformly bounded. Hence, the convergence in $L^2$ follows from the result in \eqref{eq:ConvergenceResultChiaraDomi}.
We finalise the remark by mentioning that the Boltzmann operator conserves the total mass and therefore
\begin{equation}
    f_\varepsilon(\eta_\varepsilon^{-\delta} t) - \Tilde{h}_\varepsilon(t) \xrightarrow[]{\varepsilon \to 0} 0 \quad \text{in} \quad L^\infty ([0,T];L^2(\mathbb{R}^d \times S^{d-1}_\abs{v})).
\end{equation}
\\[1cm]
\small{\textit{Acknowledgements.} The author is funded by the Swiss National Science Foundation through the NCCR SwissMAP and the SNSF Eccellenza project PCEFP2\textunderscore181153 as well as acknowledges the support of the Swiss State Secretariat for Research and Innovation by the project P.530.1016 (AEQUA). Furthermore, the author would like to thank Chiara Saffirio for insightful discussions.}

\printbibliography

\end{document}